\documentclass[a4paper]{article}
\usepackage[margin=1in]{geometry}
\usepackage{enumerate}
\usepackage{amssymb}
\usepackage{amsthm}
\usepackage{amsmath}
\usepackage{graphicx}
\usepackage{stmaryrd}

\usepackage[algoruled,linesnumbered,algo2e,vlined]{algorithm2e}
\SetArgSty{textrm}

\newcommand{\SA}{\mathit{SA}}
\newcommand{\ISA}{\mathit{ISA}}

\newcommand{\lcp}{\mathit{lcp}}

\newcommand{\LyndonTree}{\mathit{LTree}}
\newcommand{\LCA}{\mathit{lca}}
\newcommand{\RMQ}{\mathit{rmq}}
\newcommand{\ignore}[1]{}
\newcommand{\lexorderl}[1]{\prec_{#1}}

\newcommand{\Lroot}[1]{\textsf{L}-root}

\newcommand{\runs}{\mathit{Runs}}
\newcommand{\begset}{\mathit{Beg}}
\newcommand{\longestLyndonl}[2]{l_{#2}(#1)}
\newcommand{\exrun}{\mathit{exrun}}

\newtheorem{theorem}{Theorem}
\newtheorem{definition}[theorem]{Definition}
\newtheorem{lemma}[theorem]{Lemma}

\title{
  The ``Runs'' Theorem\footnote{A preliminary version of this paper has appeared in~\cite{bannai15:_lyndon}.}
}
\author{
  Hideo~Bannai$^1$\quad
  Tomohiro~I$^2$\quad
  Shunsuke~Inenaga$^1$\quad
  Yuto~Nakashima$^1$\quad\\
  Masayuki~Takeda$^1$\quad
  Kazuya~Tsuruta$^1$\\
  {$^1$ Department of Informatics, Kyushu University}\\
  {\texttt{\{bannai,inenaga,takeda,yuto.nakashima\}@inf.kyushu-u.ac.jp}}\\
  {$^2$ Department of Computer Science, TU Dortmund, Germany}\\
  {\texttt{tomohiro.i@cs.tu-dortmund.de}}
}

\date{}
\begin{document}
\maketitle

\begin{abstract}
  We give a new characterization of maximal repetitions (or runs) in strings based on Lyndon words.
  The characterization leads to a proof of what was known as the ``runs'' conjecture
  (Kolpakov \& Kucherov (FOCS '99)), which states that
  the maximum number of runs $\rho(n)$ in a string of length $n$ is less than $n$.
  The proof is remarkably simple, considering the numerous endeavors to tackle this problem 
  in the last 15 years,
  and significantly improves our understanding of how runs can occur in strings.
  In addition, we obtain an upper bound of $3n$ for the maximum sum of exponents
  $\sigma(n)$ of runs in a string of length $n$,
  improving on the best known bound of $4.1n$ by Crochemore et al. (JDA 2012),
  as well as other improved bounds on related problems.
  The characterization also gives rise to a new, conceptually simple linear-time
  algorithm for computing all the runs in a string.
  A notable characteristic of our algorithm is that, unlike all existing linear-time algorithms,
  it does {\em not} utilize the Lempel-Ziv factorization of the string.
  We also establish a relationship between runs and nodes of the Lyndon tree,
  which gives a simple optimal solution to the
  2-Period Query problem that was recently solved by Kociumaka et al. (SODA 2015).
\end{abstract}

\section{Introduction}
Repetitions in strings are one of the most basic
and well studied characteristics of strings, with various theoretical 
and practical applications (See~\cite{smyth00:_repet,crochemore09:_repet,smyth13:_comput} for surveys).
In this paper, we focus on maximal repetitions, or {\em runs}.
A run is a maximal periodic sub-interval
of a string, that is at least as long as twice its smallest period.
For example, for a string $w[1..11] = \texttt{aababaababb}$,
$[1..2] = \texttt{a}^2$,
$[6..7]=\texttt{a}^2$, and
$[10..11]=\texttt{b}^2$
are runs with period $1$,
$[2..6] = (\texttt{ab})^{5/2}$ and $[7..10]=(\texttt{ab})^2$ are runs with period $2$,
$[4..9] = (\texttt{aba})^2$ is a run with period $3$, and $[1..10] = (\texttt{aabab})^2$
is a run with period $5$.
Runs essentially capture all consecutive repeats of a 
substring in a string.

The most remarkable non-trivial property of runs, 
first proved by Kolpakov and Kucherov~\cite{kolpakov99:_findin_maxim_repet_word_linear_time},
is that the maximum number of runs $\rho(n)$ in a string of length $n$,
is in fact linear in $n$.
Although their proof did not give a specific constant factor, 
it was conjectured that $\rho(n) < n$.
In order to further understand the combinatorial structure of runs in strings,
this ``runs conjecture'' has, since then, become the focus of many investigations.
The first explicit constant was given
by Rytter~\cite{rytter06:_number_runs_strin},
where he showed $\rho(n) < 5n$.
This was subsequently improved to $\rho(n) < 3.48n$
by Puglisi et al.~\cite{puglisi06:_how}
with a more detailed analysis using the same approach.
Crochemore and Ilie~\cite{crochemore08:_maxim} 
further reduced the bound to $\rho(n) < 1.6n$, 
and showed how better bounds could be obtained by computer verification.
Based on this approach, Giraud proved
$\rho(n) < 1.52n$~\cite{giraud08:_not_so_many_runs_strin}
and later $\rho(n) < 1.29n$~\cite{giraud09:_asymp},
but only for binary strings.
The best known upper bound is $\rho(n) < 1.029n$ obtained by
intense computer verification (almost 3 CPU years)~\cite{crochemore11}, 
based on~\cite{crochemore08:_maxim}.
On the other hand, a lower bound of $\rho(n) \geq 0.927n$
was shown by Franek et al.~\cite{franek08}.
Although this bound was first conjectured to be optimal, the bound was later improved by
Matsubara et al.~\cite{matsubara08:_new} to $\rho(n) \geq 0.944565n$.
The best known lower bound is $\rho(n) \geq 0.944575712n$ 
by Simpson~\cite{simpson10:_modif_padov}.
While the conjecture was very close to being proved, 
all of the previous linear upper bound proofs
are based on heavy application of the periodicity lemma by Fine and Wilf~\cite{fine65:_uniquen},
and are known to be very technical, 
which seems to indicate that we still do not yet have a good understanding of how runs can be contained in strings.
For example, the proof for $\rho(n) < 1.6n$
by Crochemore and Ilie~\cite{crochemore08:_maxim} 
required consideration of at least 61 cases (Table 2 of~\cite{crochemore08:_maxim})
in order to bound the number of runs with period at most $9$ by $n$.

In this paper, we give new insights into this difficult problem, significantly 
improving our understanding of the structure of runs in strings.
Our study of runs is based on combinatorics of Lyndon words~\cite{lyndon54:_burnside}.
A Lyndon word is a string that is lexicographically smaller than all of 
its proper suffixes. Despite the simplicity of its definition,
Lyndon words have many deep and interesting combinatorial properties~\cite{Lothaire83}
and have been applied to a wide range of problems~\cite{Lothaire83,reutenauer93:_free_lie_algeb,lalonde95:_stand_lyndon_bases_lie_algeb_envel_algeb,delgrange04:_star,Chemillier04:_periodic_musical_Lyndon,dieudonne07:_circl_lyndon,kufleitner09:_bijective_BWT,BrlekLPR09,hill12:_repres_hecke_lyndon,Mucha13:_lyndon_superstring,dieudonne13:_deter}.
Lyndon words have recently been considered in the context of
runs~\cite{crochemorea12,crochemore14:_extrac}, since any run 
with period $p$ must contain a length-$p$ substring that is a Lyndon word,
called an \Lroot{} of the run.
Concerning the number of cubic runs (runs with exponent at least 3),
Crochemore et al.~\cite{crochemorea12} gave a very simple proof
that it can be no more than $0.5n$. The key observation is that,
for any given lexicographic order,
a cubic run must contain at least two consecutive occurrences of its \Lroot{},
and that the boundary position cannot be shared by consecutive \Lroot{}s of a
different cubic run.
However, this idea does not work for general runs, since,
unlike cubic runs, only one occurrence of an \Lroot{} 
for a given lexicographic order is guaranteed,
and the question of how to effectively apply Lyndon arguments to the analysis 
of the number of general runs has so far not been answered.

The contributions of this paper are summarized below:
\begin{description}
  \item[Proof of \mbox{\boldmath$\rho(n) < n$} and \mbox{\boldmath$\sigma(n) < 3n$}]
    We discover and establish a connection between the \Lroot{}s of
    runs and the longest Lyndon word starting at each position of the string.
    Based on this novel observation, 
    we give an affirmative answer to the runs conjecture.
    The proof is remarkably simple.

    Based on the same observation, we obtain a bound of $3n$ for the maximum sum
    of exponents $\sigma(n)$ of runs in a string of length $n$.
    The best known bound was $4.1n$ by Crochemore et al.~\cite{crochemore12},
    whose arguments were based on the bound of $\rho(n) < 1.029n$.
    We note that plugging-in $\rho(n) < n$ into their proof still only gives a bound of $4n$.

    For higher exponent runs with exponent at least $k\geq 2$, we prove a bound of
    $\rho_k(n) < n/(k-1)$ and $\sigma_k(n) < n(k+1)/(k-1)$, where
    $\rho_k(n)$ is the maximum number of runs with exponent at least $k$ in a string of length $n$,
    and $\sigma_k(n)$ is the maximum sum of exponents of runs with exponent at least $k$ in a string
    of length $n$. For $k = 3$, this yields $\sigma_3(n) < 2n$
    which improves on the bound of $2.5n$ by Crochemore et al.~\cite{crochemore12}.

    We also prove conjectured bounds of $\rho(n,d) \leq n - d$ and
    if $n > 2d$, $\rho(n,d) \leq n - d - 1$ ~\cite{deza14},
    where $\rho(n,d)$ is the
    maximum number of runs in a string of length $n$ that contains exactly
    $d$ distinct symbols\footnote{We note that
      Deza and Franek have independently and simultaneously proved similar bounds~\cite{deza15},
      based on our proof of the runs conjecture in an earlier version of this paper.
    }.
  \item[Linear-time computation of all runs without Lempel-Ziv parsing]
    We give a novel, conceptually simple linear-time algorithm for
    computing all runs contained in a string, based on the proof of $\rho(n) < n$.
    The first linear-time algorithm for computing all runs, proposed by Kolpakov 
    and Kucherov~\cite{kolpakov99:_findin_maxim_repet_word_linear_time}, 
    relies on the computation of the Lempel-Ziv parsing~\cite{LZ77} of the string.
    All other existing linear-time algorithms basically follow their algorithm,
    but focus on more efficient computation of the parsing, which is the bottleneck
    ~\cite{chen07:_fast_pract_algor_comput_all_runs_strin,crochemore08:_comput_longes_previous_factor}.
    Our algorithm is the first linear-time algorithm which does {\em not} rely on the Lempel-Ziv parsing of the string, 
    and thus may help pave the way to more efficient algorithms for computing all runs
    in the string~\cite{smyth14:_large}.
  \item[Runs and Lyndon trees]
    We also establish a relationship between \Lroot{}s of runs in a string
    and nodes of what is called the Lyndon tree of the string~\cite{barcelo90:_free_lie_algeb},
    which is a full binary tree defined by recursive standard factorization.
    We show a simple optimal solution to the 2-Period Query problem
    that was recently solved by Kociumaka et al.~\cite{Kociumaka2015IPM},
    i.e., given any interval $[i..j]$ of a string $w$ of length $n$,
    return the smallest period $p$ of $w[i..j]$ with $p \le (j-i+1)/2$,
    if such exists, in constant time with $O(n)$ preprocessing.
\end{description}

The rest of the paper is organized as follows.
In Section~\ref{sec:preliminaries} we give basic definitions.
In Section~\ref{sec:runs}, we prove that $\rho(n) < n$.
The new linear-time algorithm for computing all runs in a string is
described in~Section~\ref{sec:algorithm}.
Section~\ref{sec:runsandlyndontree} describes the relation between
runs and Lyndon trees, as well as our new solution for the 2-Period Query problem.
Finally, Section~\ref{sec:conclusions} concludes the paper.

\section{Preliminaries}
\label{sec:preliminaries}
Let $\Sigma$ be an ordered finite {\em alphabet}.
An element of $\Sigma^*$ is called a {\em string}.
The length of a string $s$ is denoted by $|s|$. 
The empty string $\varepsilon$ is a string of length 0.
For a string $s = xyz$, $x$, $y$ and $z$ are called
a \emph{prefix}, \emph{substring}, and \emph{suffix} of $s$, respectively.
A prefix (resp. suffix) $x$ of $s$ is called a \emph{proper prefix} (resp. suffix) 
of $s$ if $x \neq s$.
The $i$-th character of a string $s$ is denoted by $s[i]$, where $1 \leq i \leq |s|$.
For a string $s$ and two integers $1 \leq i \leq j \leq |s|$, 
let $s[i..j]$ denote the substring of $s$ that begins at position $i$ and ends at
position $j$. For convenience, let $s[i..j] = \varepsilon$ when $i > j$.
An integer $p \geq 1$ is said to be a \emph{period} of 
a string $s$ if $s[i] = s[i+p]$ for all $1 \leq i \leq |s|-p$.
For any set $I$ of intervals, let $\begset(I)$ denote the set of beginning positions of intervals in $I$.

\begin{definition}[Runs]
  A triple $r = (i,j,p)$ is a run of string $w$,
  if the smallest period $p$ of $w[i..j]$
  satisfies $|w[i..j]| \geq 2p$, and the periodicity cannot be
  extended to the left or right, i.e.,
  $i = 1$ or $w[i-1]\neq w[i+p-1]$,
  and, $j = n$ or $w[j+1]\neq w[j-p+1]$.
  The rational number $\frac{j-i+1}{p}$ is called the exponent of $r$.
\end{definition}

Let $\runs(w)$ denote the set of runs of string $w$.
Denote by $\rho(n)$, the maximum number of runs that are contained in a string of length $n$,
and by $\sigma(n)$, the maximum sum of exponents of runs that are contained in a string of length $n$.

Let $\prec$ denote some total order on $\Sigma$, as well as 
the lexicographic order induced on $\Sigma^*$.

\begin{definition}[Lyndon Word~\cite{lyndon54:_burnside}]
 A non-empty string $w \in \Sigma^+$ is said to be a \emph{Lyndon word}
 with respect to $\prec$, if $w \prec u$ for any non-empty proper suffix $u$ of $w$.
\end{definition}

Note that a Lyndon word $w$ cannot have any period $p < |w|$,
since its existence would imply $w=xyx$ for some non-empty $x,y$, and $x \prec w$.

\begin{lemma}[Lemma 1.6 of~\cite{Duval83:_facrorizing_words_}]\label{lemma:duval}
  Let $w = u^ku^{\prime}a$ be a string for some Lyndon word $u$, a possibly empty proper prefix $u^{\prime}$ of $u$, 
  a positive integer $k$, and $a \in \Sigma$ with $w[|u^{\prime}|+1] \neq a$.
  If $u[|u^{\prime}|+1] \prec a$, $w$ is a Lyndon word.
  If $a \prec u[|u^{\prime}|+1]$, $u$ is the longest prefix Lyndon word of any string having a prefix $u^ku^{\prime}a$.
\end{lemma}

\begin{definition}[\Lroot{}~\cite{crochemore14:_extrac}]
  Let $r = (i,j, p)$ be a run in string $w\in\Sigma^*$.
  An interval $\lambda = [i_\lambda..j_\lambda]$ of length $p$
  is an \Lroot{} of $r$ with respect to $\prec$
  if $i \leq i_\lambda \leq j_\lambda \leq j$ and
  $w[i_\lambda..j_\lambda]$ is a Lyndon word with respect to $\prec$.
\end{definition}
It is easy to see that for any run and lexicographic order $\prec$,
there exists at least one \Lroot{} with respect to $\prec$.

\section{The Runs Theorem}
\label{sec:runs}
Since any string over a unary alphabet can only have at most one run,
we assume a non-unary alphabet $\Sigma$.
Furthermore,
we consider lexicographic orders on strings over $\Sigma$,
induced by an arbitrary pair of total orders $\lexorderl{0}$, $\lexorderl{1}$ on $\Sigma$
such that for any pair of characters $a,b \in \Sigma, a\lexorderl{0}b \Leftrightarrow b\lexorderl{1}a$.
For $\ell\in\{0,1\}$, let $\overline{\ell} = 1 - \ell$.
For any string $w \in \Sigma^*$, let $\hat{w} = w\$$, where $\$\not\in\Sigma$
is a special character that satisfies $\$ \lexorderl{0} a$ (and thus $a \lexorderl{1}\$$)
for any $a\in\Sigma$.

\begin{definition}
  For any string $w$ and position $i~(1 \leq i\leq |w|)$,
  let 
  $\longestLyndonl{i}{\ell} = [i..j]$ where
  $j = \max\{ j^\prime \mid \hat{w}[i..j^\prime] \mbox{ is a Lyndon word with respect to} \lexorderl{\ell}\}$
\end{definition}

\begin{lemma}\label{lemma:onlyoneislongerthanone}
  For any string $w$ of length $n$ and position $i~(1 \leq i \leq |w|)$,
  we have
  for a unique $\ell\in\{0,1\}$ that
  $\longestLyndonl{i}{\ell} = [i..i]$ 
  and $\longestLyndonl{i}{\overline{\ell}} = [i..j]$ where $j > i$.
\end{lemma}
\begin{proof}
  Let $k = \min\{ k^\prime \mid \hat{w}[k^\prime] \neq \hat{w}[i], k^\prime > i\}$,
  and let $\ell \in \{0,1\}$ be such that $\hat{w}[k] \lexorderl{\ell} \hat{w}[i]$.
  It follows from Lemma~\ref{lemma:duval},
  that $\longestLyndonl{i}{\ell} = [i..i]$ and
  $\longestLyndonl{i}{\overline{\ell}} = [i..j]$ for some $j \geq k > i$.
\end{proof}

\begin{lemma}\label{lemma:lrootislongest}
  Let $r = (i,j,p)$ be an arbitrary run in string $w$ of length $n$.
  Then, for a unique $\ell\in\{0,1\}$ such that $\hat{w}[j+1] \lexorderl{\ell} \hat{w}[j+1-p]$,
  any \Lroot{} $\lambda = [i_\lambda..j_\lambda]$ of $r$ with respect to $\lexorderl{\ell}$ 
  is equal to $\longestLyndonl{i_\lambda}{\ell}$.
\end{lemma}
\begin{proof}
  By the definition of $r$, $\hat{w}[j+1] \neq \hat{w}[j+1-p]$.
  Therefore, there exists a unique $\ell\in\{0,1\}$ 
  such that $\hat{w}[j+1] \lexorderl{\ell} \hat{w}[j+1-p]$.
  Let $[i_\lambda..j_\lambda]$ be an \Lroot{} of $r$ with respect to $\lexorderl{\ell}$.
  It follows from Lemma~\ref{lemma:duval} that $[i_\lambda..j_\lambda] = \longestLyndonl{i_\lambda}{\ell}$.
\end{proof}

For any run $r = (i,j,p)$ of $w$,
let
$B_r = \{ \lambda = [i_\lambda..j_\lambda] \mid \lambda
        \mbox{ is an \Lroot{} of $r$ with respect to $\lexorderl{\ell}$},$ $i_\lambda \neq i \}$,
where $\ell\in\{0,1\}$ is such that
$\hat{w}[j+1] \lexorderl{\ell} \hat{w}[j+1-p]$,
i.e., $B_r$ is the set of all \Lroot{}s $[i_\lambda..j_\lambda]$ of $r$ 
with respect to $\lexorderl{\ell}$ such that $[i_\lambda..j_\lambda] = \longestLyndonl{i_\lambda}{\ell}$,
except for the one that starts from $i$ if it exists.
Note that $|\begset(B_r)| = |B_r| \geq \lfloor e_r - 1\rfloor \geq 1$, where $e_r$ is the exponent of $r$.

\begin{lemma}\label{lemma:exclusive}
  For any two distinct runs $r$ and $r^{\prime}$ of string $w$, 
  $\begset(B_r)\cap\begset(B_{r^{\prime}})$ is empty.
\end{lemma}
\begin{proof}
  Suppose that there exist $i\in\begset(B_r)\cap\begset(B_{r^{\prime}})$, and
  $\lambda = [i..j_{\lambda}] \in B_{r}$ and $\lambda^{\prime} = [i..j_{\lambda^{\prime}}]\in B_{r^{\prime}}$.
  Let $\ell \in \{ 0, 1\}$ be such that $\lambda = \longestLyndonl{i}{\ell}$.
  Since $\lambda \neq \lambda^{\prime}$, $\lambda^{\prime} = \longestLyndonl{i}{\overline{\ell}}$.
  By Lemma~\ref{lemma:onlyoneislongerthanone}, either $\lambda$ or $\lambda^{\prime}$ is $[i..i]$.
  Assume w.l.o.g. that $\lambda = [i..i]$
  and $j_{\lambda^{\prime}} > i$.
  Since $w[i..j_{\lambda^{\prime}}]$ is a Lyndon word, $w[i] \neq w[j_{\lambda^{\prime}}]$.
  By the definition of $B_r$ and $B_{r^{\prime}}$, the beginning positions of runs $r$ and $r^\prime$ are both less than $i$,
  which implies $w[i-1] = w[i]$ (due to $r$) and $w[i-1] = w[j_{\lambda^{\prime}}]$ (due to $r'$).
  Hence we get $w[i] = w[i-1] = w[j_{\lambda^{\prime}}]$, a contradiction.
\end{proof}

Lemma~\ref{lemma:exclusive} shows that each run $r$ can be associated with 
a disjoint set of positions $\begset(B_r)$.
Also, since $1\not\in\begset(B_r)$ for any run $r$,
$\sum_{r \in \runs(w)} |B_r| = \sum_{r \in \runs(w)} |\begset(B_r)| \leq |w|-1$ holds.
Therefore, we obtain the following results.

\begin{theorem}\label{theorem:rho}
  $\rho(n) < n$.
\end{theorem}
\begin{proof}
  Consider string $w$ of length $n$.
  Since $|B_r| \geq 1$ for any $r \in \runs(w)$,
  it follows from Lemma~\ref{lemma:exclusive} that $|\runs(w)| \leq \sum_{r \in \runs(w)} |B_r| \leq n-1$.
\end{proof}

\begin{theorem}
  $\sigma(n) \leq 3n-3$.
\end{theorem}
\begin{proof}
  Consider string $w$ of length $n$.
  Let $e_r$ denote the exponent of run $r$.
  Since $|B_r| \geq \lfloor e_r - 1 \rfloor > e_r - 2$ for any $r \in \runs(w)$, 
  it follows from Lemma~\ref{lemma:exclusive} that 
  $\sum_{r \in \runs(w)} (e_r - 2) < \sum_{r \in \runs(w)} \lfloor e_r - 1 \rfloor \leq \sum_{r \in \runs(w)} |B_r| \leq n-1$.
  Using $|\runs(w)| \leq n-1$ from Theorem~\ref{theorem:rho}, we get $\sum_{r \in \runs(w)} e_r < n + 2|\runs(w)| - 1 \leq 3n-3$.
\end{proof}

\subsection{Higher Exponent Runs}
Let $\runs_k(w)$ denote the set of runs of string $w$ with exponent at least $k$, 
$\rho_k(n)$ the maximum number of runs with exponent at least $k$ in a string of length $n$,
and $\sigma_k(n)$ the maximum sum of exponents of runs with exponent at least $k$ in a string of length $n$.
Crochemore et al.~\cite{crochemore12} have shown a bound of
$2.5n$ for $\sigma_3(n)$. 
Below, we prove a tighter bound, and show bounds for general integer $k$ as well.

\begin{theorem}
  $\rho_k(n) < n/(k-1)$, $\sigma_k(n) < n(k+1)/(k-1)$.
\end{theorem}
\begin{proof}
  Notice that for any run $r$ with exponent at least $k$,
  $|B_r| \geq \lfloor e_r - 1 \rfloor \geq k-1$.
  Therefore, 
  $|\runs_k(w)| \leq \sum_{r\in\runs_k(w)}|B_r|/(k-1) \leq n/(k-1)$.
  Also,
  $\sum_{r\in\runs_k(w)} e_r 
  = \sum_{r\in\runs_k(w)}(e_r-2) + 2|\runs_k(w)|
  \leq \sum_{r\in\runs_k(w)} |B_r| + 2n/(k-1)
  < n + 2n/(k-1) = n(k+1)/(k-1)$.
\end{proof}

\subsection{Runs with $d$ distinct symbols}
Let $\rho(n,d)$ denote the maximum number of runs in a string of length $n$ 
that contains exactly $d$ distinct symbols.
We prove the following bounds conjectured in~\cite{deza14}.
\begin{theorem}
  $\rho(n,d) \leq n - d$. 
  Furthermore, if $n > 2d$, then $\rho(n,d) \leq n - d - 1$.
\end{theorem}
\begin{proof}
  Let $\Sigma = \{ c_1, \ldots, c_d\}$.
  First, we show $\rho(n,d) \leq n - d$. 
  For any character $c_k \in \Sigma$, let $i_k$ denote its last occurrence,
  i.e. $i_k = \max\{ i \mid w[i] = c_k, 1 \leq i \leq n \}$.
  Choose the pair of total orders $\lexorderl{0},\lexorderl{1}$ on $\Sigma$, so that
  for any $1 \leq k,k^\prime \leq d$,
  $c_{k^\prime} \lexorderl{0} c_{k}\Leftrightarrow c_k \lexorderl{1} c_{k^\prime} \Leftrightarrow i_k < i_{k^\prime}$.
  Also, let $i^\prime_k = \min \{ i \leq i_k \mid w[i..i_k] = c_k^{i_k-i+1} \}$.
  Then, for any $1 \leq k\leq d$,
  since $c_k = w[i^\prime_k] = \cdots = w[i_k]$ is smaller than any
  character in $\hat{w}[i_k+1..n+1]$ with respect to
  $\lexorderl{1}$, we have that
  $\longestLyndonl{i^\prime_k}{1} = [i^\prime_k..n+1]$, and from Lemma~\ref{lemma:onlyoneislongerthanone},
  $\longestLyndonl{i^\prime_k}{0} = [i^\prime_k..i^\prime_k]$.
  Since $\hat{w}[i^\prime_k..n+1]$ includes the symbol $\$$ which does not occur elsewhere in $\hat{w}$,
  $[i^\prime_k..n+1]$ cannot be an \Lroot{} of a run.
  On the other hand, if $[i^\prime_k..i^\prime_k]$ is an \Lroot{} of some run,
  then by definition of $i^\prime_k$, the run must start at $i^\prime_k$.
  Therefore, neither $\longestLyndonl{i^\prime_k}{0}$ nor $\longestLyndonl{i^\prime_k}{1}$
  can be included in $\cup_{r\in\runs(w)} B_r$ and thus,
  $i^\prime_k \not\in \cup_{r\in\runs(w)} \begset(B_r)$.
  Noticing that $w[i^\prime_k] = c_k$, we have that $i^\prime_k$ is different for each
  $1\leq k\leq d$, and therefore, $\rho(n,d) \leq n - d$.

  Next, we prove $\rho(n,d) \leq n - d - 1$ for $n > 2d$.
  Since $1 \not\in \cup_{r\in\runs(w)} \begset(B_r)$, if $i^\prime_k > 1$ for all $k$, then
  $\runs(w) \leq n - d - 1$.
  Therefore, we can assume $i^\prime_1 = 1$, which means that
  $w[1..i_1] = c_1^{i_1}$, and $w[i_1+1..n]$ does not contain an occurrence of $c_1$.
  Thus, any position in $w[1..i_1]$ can only be part of a single run $(1,i_1,1)$ if $i_1>1$,
  or of none if $i_1=1$.
  If $i_1 > 1$, we have from the first statement that
  $\runs(w) \leq 1 + \rho(n - i_1, d-1) \leq 1 + (n - i_1) - (d-1) = n - d - (i_1 - 2)$.
  Since $\runs(w) \leq n - d - 1$ for $i_1 \geq 3$, we assume that $i_1 \leq 2$.
  We prove the statement by induction on $d$.
  For $d = 1$, we have that $\rho(n,1) \leq 1$, and thus $\rho(n,1) \leq n - d - 1$
  for any $n > 2$, and the statement holds.
  Suppose the statement holds for any $d^\prime < d$, i.e.,
  for any $d^\prime < d$, if $n > 2d^\prime$ then $\rho(n,d^\prime) \leq n - d^\prime - 1$.
  If $i_1 = 1$, then, since $(n-1) > 2(d-1)$,
  we have $\runs(w) \leq \rho(n-1,d-1) \leq (n-1) - (d-1) - 1 \leq n - d - 1$.
  If $i_1 = 2$, then, again since $(n-2) > 2(d-1)$,
  we have $\runs(w) \leq 1 + \rho(n-2,d-1) \leq (n-2) - (d-1) \leq n - d - 1$.
  Thus, the statement holds.
\end{proof}

This leads to a slightly better bound of $\rho(n)$ compared to Theorem~\ref{theorem:rho}, i.e.,
$\rho(n) \leq n-3$ for $n > 4$, since $\rho(n,1) \leq 1$.

\section{New Linear-Time Algorithm for Computing All Runs}
\label{sec:algorithm}
In this section, we describe our new linear-time algorithm for computing all runs 
in a given string $w$ of length $n$.
As there is a lower bound of $\Omega(n\log n)$ time for any algorithm that
is based on character comparisons~\cite{main84:_o_algor_findin_all_repet_strin},
we assume an integer alphabet, i.e. $\Sigma = \{1,...,n^c\}$ for some constant $c$.
Let $L = \{ \longestLyndonl{i}{\ell} \mid \ell\in\{0,1\},  1 \leq i \leq n\}$.
From Lemma~\ref{lemma:lrootislongest}, 
we know that for any run $r$, $L$ contains an \Lroot{} of $r$.
Our new algorithm 
(1) 
computes the set $L$ in linear time,
and
(2) for each element $\longestLyndonl{i}{\ell} \in L$,
checks if it is equal to $\arg_{[i..j]\in B_r}\min i$ for some run,
and if so determine the run, in constant time,
therefore achieving linear time.
Below are the algorithmic tools used in our algorithm.

\begin{definition}[Suffix Array/Inverse Suffix Array~\cite{manber93:_suffix}]
  The suffix array $\SA_w[1..n]$ of a string $w$ of length $n$,
  is an array of integers such that $\SA_w[i]=j$ indicates that $w[j..n]$ is the 
  lexicographically $i$th smallest suffix of $w$.
  The inverse suffix array $\ISA_w[1..n]$ is an array of integers such 
  that $\ISA_w[\SA_w[i]] = i$.
\end{definition}
\begin{theorem}[Suffix Array/Inverse Suffix Array~\cite{Kim03,KoAluru05:_space_effi_linear_const_suff_arr,karkkainen06:_linear_suff_arr}]
\label{theorem:saisa}
  The suffix array and inverse suffix array of a string over an integer alphabet
  can be computed in linear time.
\end{theorem}

\begin{theorem}[Range Minimum Query~\cite{BenderF00}]\label{theorem:rmq}
  An array $A[1..n]$ of integers can be preprocessed in linear time so that for any 
  $1 \leq i \leq j \leq n$,
  $\RMQ_A(i,j) = \arg\min_{i \leq k\leq j}\{ A[k] \}$ can be computed in linear time.
\end{theorem}

\begin{theorem}[Longest Common Extension Query~(e.g.,~\cite{fischer06:_theor_pract_improv_rmq_probl})]\label{theorem:lce}
  A string $w$ over an integer alphabet can be preprocessed in linear time, 
  so that for any $1 \leq i \leq j \leq |w|$,
  $|\lcp(w[i..|w|], w[j..|w|])|$  can be answered in constant time.
\end{theorem}

\subsection{Linear-Time Computation of \mbox{\boldmath$\longestLyndonl{i}{\ell}$}}
Algorithm~\ref{algo:LyndonTreePseudoCode} shows a pseudo-code of a linear-time algorithm 
that computes $\longestLyndonl{i}{\ell}$ for some $\ell\in\{0,1\}$,
in a right-to-left scan of $w$ using a stack.
The correctness of the algorithm can be seen from the following facts.
\begin{lemma}[Theorem (1.4) of~\cite{ChenFL58:_lyndon_factorization_}]\label{lemma:concatlyndon}
  For any Lyndon words $u$ and $v$ such that $u\prec v$, $uv$ is a Lyndon word.
\end{lemma}
\begin{lemma}[Lyndon Factorization and Longest Lyndon Prefix~\cite{ChenFL58:_lyndon_factorization_,Duval83:_facrorizing_words_}]\label{lemma:lyndonfactorization}
  Any string $w$ can be decomposed into a unique sequence $f_1\cdots f_m$
  of lexicographically non-increasing Lyndon words, called the Lyndon factorization of $w$.
  Furthermore, each factor $f_i~(1\leq i\leq m)$ is the longest Lyndon 
  word that is a prefix of $f_i\cdots f_m$.
\end{lemma}

From Lemma~\ref{lemma:concatlyndon},
it is easy to see that at the end of each loop for $i$ in the algorithm, the stack $S$ contains
a lexicographically non-increasing list of Lyndon words that decomposes $w[i..n]$, 
and thus is the Lyndon factorization of $w[i..n]$.
The top element of the stack is the first Lyndon factor,
and therefore, from Lemma~\ref{lemma:lyndonfactorization},
is the longest Lyndon word that starts at position $i$.

The lexicographic comparison of Line~\ref{algo:lexcomparison} can be
performed in constant time by utilizing $\ISA_{\hat{w}}$, i.e., the lexicographic order of 
the suffix of $\hat{w}$ starting at the same position.
Consider a Lyndon word $f_0$ starting at position $i$,
and the Lyndon factorization $f_1\cdots f_m$ of $\hat{w}[i_v..n+1]$, where $i_v = i+|f_0|$.
If $f_0 \prec f_1$, then, $f = f_0f_1$ is a Lyndon word from Lemma~\ref{lemma:concatlyndon}.
Therefore, $f[1..|f_1|] \prec f_1$ and thus $f_0\cdots f_m \prec f_1\cdots f_m~(\ISA_{\hat{w}}[i] < \ISA_{\hat{w}}[i_v])$.
If $f_1 \preceq f_0$, then $f_0\cdots f_m$ is a Lyndon factorization of $\hat{w}[i..n+1]$.
It follows from Lemma~\ref{lemma:duval} that $f_1\cdots f_m \prec f_0\cdots f_m~(\ISA_{\hat{w}}[i_v] < \ISA_{\hat{w}}[i])$,
since $f_0$ must be the longest Lyndon prefix of $\hat{w}[i..n+1]$.
Therefore $f_0\prec f_1 \iff \ISA[i] < \ISA[i_v]$.

We note that the intervals constructed during the algorithm correspond to
nodes of what is called the Lyndon tree~\cite{barcelo90:_free_lie_algeb}, described in Section~\ref{sec:runsandlyndontree}.
Hohlweg and Reutenauer~\cite{hohlweg03:_lyndon} showed that the
Lyndon tree can be constructed in linear time given $\ISA$, by showing
that the Cartesian tree~\cite{vuillemin80,gabow84:_scalin} of
the subarray $\ISA[2..n]$ coincides with the internal nodes of the Lyndon tree.
Algorithm~\ref{algo:LyndonTreePseudoCode} is, in essence, an implementation of the same idea.

\begin{algorithm2e}
  \caption{Computing $\longestLyndonl{i}{\ell}$ in linear time for all $i$.}
  \label{algo:LyndonTreePseudoCode}
  \KwIn{String $w$ of length $n$}
  $S \leftarrow$ new stack with element $(n+1,n+1)$\;
  \For{$i = n$ \textbf{downto} $1$\label{algo:forloop}}{
    $j \leftarrow i$\;
    \While{$S$.size() $> 1$\label{algo:whileloop}}{
      $(i_v,j_v) \leftarrow S$.top() \;
      \lIf(\tcp*[f]{$O(1)$ from $\ISA_{\hat{w}}[i],\ISA_{\hat{w}}[i_v]$}){%
        \textbf{not} $\hat{w}[i..j]\lexorderl{\ell} \hat{w}[i_v..j_v]$\label{algo:lexcomparison}
      }{exit \textbf{while} loop }
      $j \leftarrow j_v$ \tcp*{$\hat{w}[i..j]$ is Lyndon w.r.t.$\lexorderl{\ell}$}
      $S$.pop()\label{algo:pop}\;
    }
    $S$.push($(i,j)$)\label{algo:stackpush}\label{algo:forloopend}\;
    $\longestLyndonl{i}{\ell} \leftarrow [i..j]$\;
  }
\end{algorithm2e}

\subsection{Computing All Runs of \mbox{\boldmath$w$} from \mbox{\boldmath$\longestLyndonl{i}{\ell}$}}
\label{subsec:runsfromlongestlyndon}
Consider a candidate interval $\longestLyndonl{i}{\ell} = [i..j] \in L$.
Let $w[i^\prime..i-1]$ be the longest common suffix of $w[1..i-1]$ and $w[1..j]$,
and let $w[j+1..j^\prime]$ be the longest common prefix of $w[i..n]$ and $w[j+1..n]$.
It is easy to see that $[i..j] = \arg_{[i..j]\in B_r}\min i$
of run $r = (i^\prime,j^\prime,p)$,
if and only if $p = j-i+1$, $|w[i^\prime..j^\prime]| \geq 2p$, and $i^\prime < i \leq i^\prime + p$.
Using 
Theorem~\ref{theorem:lce}, we can compute $j^\prime$ in constant time per query and linear-time preprocessing.
If we consider LCE queries on the reverse string,
we can query the length of the longest common suffix between two prefixes of $w$.
Thus, $i^\prime$ can also be computed in constant time per query
and linear-time preprocessing. 

\section{Runs and Lyndon Trees}
\label{sec:runsandlyndontree}

In this section, we characterize runs in strings using Lyndon trees.

\begin{definition}[Standard Factorization~\cite{ChenFL58:_lyndon_factorization_,Lothaire83}] 
  The \emph{standard factorization} of a Lyndon word $w$ with $|w| \geq 2$
  is an ordered pair $(u, v)$ of Lyndon words $u,v$
  such that $w = uv$ and $v$ is the lexicographically smallest proper suffix of $w$.
\end{definition}

It can be shown that for any Lyndon word $w$ longer than 1, the standard factorization $(u,v)$ of $w$ always exists.
The Lyndon tree of a Lyndon word $w$, defined below,
is the full binary tree defined by recursive standard factorization of $w$.

\begin{definition}[Lyndon Tree~\cite{barcelo90:_free_lie_algeb}]
  The \emph{Lyndon tree} of a Lyndon word $w$, denoted $\LyndonTree(w)$, is
  an ordered full binary tree defined recursively as follows:
  \begin{itemize}
  \item if $|w| = 1$, then $\LyndonTree(w)$ consists of a single node labeled by $w$;
  \item if $|w| \geq 2$, then the root of $\LyndonTree(w)$, labeled by $w$,
    has left child $\LyndonTree(u)$ and right child $\LyndonTree(v)$,
    where $(u, v)$ is the standard factorization of $w$.
  \end{itemize}
\end{definition}
Each node $\alpha$ in $\LyndonTree(w)$ can be represented by an interval $[i..j]~(1\leq i\leq j\leq |w|)$ of $w$,
and we say that the interval $[i..j]$ corresponds to a node in $\LyndonTree(w)$.
Let $\LCA([i..j])$ denote the lowest node in $\LyndonTree(w)$ 
containing all leaves corresponding to
positions in $[i..j]$ in its subtree, or equivalently,
the lowest common ancestor of leaves at position $i$ and $j$.
Note that an interval $[i..j]$ corresponds to a node in the Lyndon tree,
iff $\LCA([i..j]) = [i..j]$.
Figure~\ref{fig:lyndonTree} shows an example of a Lyndon tree for the
Lyndon word \texttt{aababaababb}.

\begin{figure}[t]
  \centerline{\includegraphics[width=0.4\textwidth]{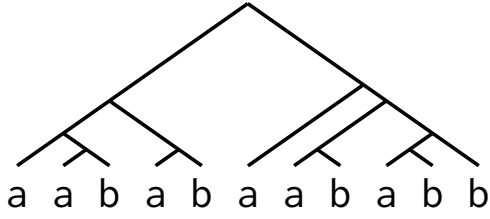}}
  \caption{A Lyndon tree for the Lyndon word
    \texttt{aababaababb}.
  }
  \label{fig:lyndonTree}
\end{figure}

We first show a simple yet powerful lemma characterizing 
Lyndon substrings of a Lyndon word, in terms of the Lyndon tree.
\begin{lemma}\label{lemma:lyndonlca}
  Let $w$ be a Lyndon word. For any interval $[i..j]$,
  if $w[i..j]$ is a Lyndon word, 
  then the node $\alpha = \LCA([i..j]) = [i_\alpha..j_\alpha]$
  in $\LyndonTree(w)$ satisfies $i_\alpha = i \leq j \leq j_\alpha$.
\end{lemma}
\begin{proof}
  If $i = j$, then $\alpha$ is a leaf node and corresponds to $[i..i]$.
  If $i < j$, $\alpha$ is an internal node.
  Let $\beta = [i_\alpha...j^\prime]$ and
  $\gamma = [j^\prime+1...j_\alpha]$ respectively be the
  left and right children of $\alpha$.
  By definition of $\LCA$,
  we have that $i_\alpha \leq i\leq j^\prime \leq j^\prime+1\leq j\leq j_\alpha$, and
  for some strings $u,v\in\Sigma^*$ and $x,y\in\Sigma^+$, we have that 
  $w[i..j] = xy$, $w[i_\alpha..j^\prime] = ux$, $w[j^\prime+1..j_\alpha] = yv$, and $(ux,yv)$
  is the standard factorization of $w[i_\alpha..j_\alpha] = uxyv$.
  Since $xy$ is a Lyndon word, $xy\prec y$, and therefore $xyv \prec yv$.
  However, if $u\neq\varepsilon$, 
  this contradicts that $yv$ is the lexicographically smallest proper suffix of $uxyv$. 
  Thus, $u$ must be empty, and $i_\alpha = i$.
\end{proof}

A node is called a \emph{left node} (resp. \emph{right node}) if it is the left (resp. right) child of its parent.
The next lemma is a simple consequence of Lemma~\ref{lemma:lyndonlca} yet also gives an important characterization.
\begin{lemma}\label{lemma:iffConditionOfRightNode}
  Let $w$ be a Lyndon word. For any interval $[i..j]$ except for $[1..|w|]$,
  $[i..j]$ corresponds to a right node of the Lyndon tree iff $w[i..j]$ is the longest Lyndon word that starts from $i$.
\end{lemma}
\begin{proof}
  Suppose $w[i..j]$ is the longest Lyndon word that starts from $i$.
  For any $j' > j$, $w[i..j']$ is not a Lyndon word and thus 
  $[i..j']$ cannot be a node in $\LyndonTree(w)$.
  Hence, it is clear from Lemma~\ref{lemma:lyndonlca} that 
  $\LCA([i..j]) = [i..j]$ and it is a right node.
  On the other hand, suppose $w[i..j]$ is not the longest Lyndon word
  that starts from $i$.
  Then, there exists a Lyndon word $w[i..j']$ for some $j' > j$.
  Since there is a node $\LCA([i..j']) = [i..j'']$ with $j'' \geq j' >j$
  due to Lemma~\ref{lemma:lyndonlca}, 
  it is easy to see that $[i..j]$ cannot be a right node.
  Note that $[i..j]$ may not correspond to a node,
  but if it does, it must be a left node.
\end{proof}

Now consider again the two total orders $\lexorderl{0}$, $\lexorderl{1}$ on $\Sigma$.
Let $w$ be an arbitrary string of length $n$ and let $w_0 = \#_0 w \$$ and $w_1 = \#_1 w \$$ where
$\#_0,\#_1\not\in\Sigma\cup\{\$\}$ are special characters that are
respectively lexicographically smaller than any other character in $\Sigma\cup\{\$\}$,
with respect to $\lexorderl{0}$ and $\lexorderl{1}$.
Thus, $\#_0 \lexorderl{0} \$ \lexorderl{0} a$ and $\#_1 \lexorderl{1} a \lexorderl{1} \$$ for any $a \in \Sigma$.
For technical reasons, we assume that positions in $w_0$ and $w_1$ will
start from $0$ rather than $1$, in order to keep in sync with positions in $w$,
i.e., so that for any $1\leq i\leq |w|$, $w[i] = w_0[i]=w_1[i]$.
Note that $w_\ell~(\ell\in\{0,1\})$  is a Lyndon word with respect to $\lexorderl{\ell}$,
and let $\LyndonTree_\ell(w)$ denote the Lyndon tree of $w_\ell$, with respect to $\lexorderl{\ell}$.
Also, $\LCA_\ell([i..j])$  will denote $\LCA([i..j])$ in $\LyndonTree_\ell(w)$.

\begin{lemma}\label{lemma:build_lynontree}
  Given a string $w$ of length $n$, $\LyndonTree_0(w)$ and $\LyndonTree_1(w)$ can be constructed in $O(n)$ time and space.
\end{lemma}

The next lemma immediately follows from Lemmas~\ref{lemma:lrootislongest} and~\ref{lemma:iffConditionOfRightNode}.
\begin{lemma}\label{lemma:run-node}
  Let $r = (i,j,p)$ be an arbitrary run in string $w$ of length $n$.
  Then, for a unique $\ell\in\{0,1\}$ such that $w_\ell[j+1] \lexorderl{\ell} w_\ell[j+1-p]$,
  any \Lroot{} of $r$ with respect to $\lexorderl{\ell}$ is a right node of $\LyndonTree_\ell(w)$.
\end{lemma}

In light of Lemma~\ref{lemma:run-node},
we have a structural view of the runs in a string $w$
by two trees $\LyndonTree_0(w)$ and $\LyndonTree_1(w)$.
This can be a powerful tool for algorithms and data structures employing subrepetitions in a string.
In the next subsection, we exhibit an application to a data structure for 2-Period Queries.

\subsection{Application to 2-Period Queries}\label{sec:2-period}
The 2-Period Query problem is to preprocess a string $w$ to support the following queries efficiently:
Given any interval $[i..j]$ of $w$, return the smallest period $p$ of $w[i..j]$ with $p \le (j-i+1)/2$, if such exists.
The 2-Period Query problem is tightly related to the runs in $w$:
Let $\exrun([i..j])$ denote a run $(i', j', p')$ such that $i' \le i, j \le j'$ and $p' \le (j-i+1)/2$ if such exists.
Note that due to the periodicity lemma~\cite{fine65:_uniquen}, such a run, if it exists, is unique and $p'$ is the smallest period of $w[i..j]$.
Therefore, a 2-Period Query with interval $[i..j]$ reduces to searching for $\exrun(i, j)$.

An optimal solution to the 2-Period Query problem was recently 
proposed in~\cite{Kociumaka2015IPM} as a by-product of their algorithm for internal pattern matching.
Their solution introduces a notion of $k$-runs 
in which a run is distributed to one or more sets of runs satisfying some conditions.
We propose another optimal yet simpler solution using Lyndon trees.

\begin{theorem}
  For any string $w$ of length $n$,
  there is a data structure of $O(n)$ space that supports 2-Period Queries in $O(1)$ time.
  The data structure can be built in $O(n)$ time.
\end{theorem}
\begin{proof}
  We construct $\LyndonTree_0(w)$ and $\LyndonTree_1(w)$
  in $O(n)$ time and space using Lemma~\ref{lemma:build_lynontree}.
  At the same time, we compute the runs in $w$ and 
  associate every node corresponding to an \Lroot{} of a run with the information of the run,
  which can be done in $O(n)$ total time as mentioned in Section \ref{sec:algorithm}.
  We also augment these trees with data structures in $O(n)$ time and space so that
  lowest common ancestor (LCA) queries can be answered in $O(1)$ time~\cite{BenderF00}.
  Given a query with interval $[i..j]$,
  our algorithm computes $\alpha_0 = \LCA_0([i..\lceil (i+j)/2 \rceil])$ and
  $\alpha_1 = \LCA_1([i..\lceil (i+j)/2 \rceil])$,
  and check their right children.
    
  Suppose that $r = \exrun([i..j]) = (i', j', p')$ exists.
  Let $\ell\in\{0,1\}$ with $w_\ell[j'+1] \lexorderl{\ell} w_\ell[j'+1-p']$.
  Since the period $p'$ of $r$ is at most $\lfloor (j-i+1)/2 \rfloor$,
  we have that
  $i \leq \lceil (i+j)/2 \rceil - p' < \lceil (i+j)/2 \rceil + p' -1 \leq j$.
  Thus, there exists an \Lroot{} $\lambda$ of $r$ with respect to $\lexorderl{\ell}$
  that contains position $\lceil (i+j)/2 \rceil$.
  By Lemma~\ref{lemma:run-node}, $\lambda$ is a right node.
  Moreover, $\alpha_\ell$ is an ancestor of $\lambda$ since $\lambda$ does not contain position $i$
  while both contain position $\lceil (i+j)/2 \rceil$.
  We claim that the right child of $\alpha_\ell$ is $\lambda$.
  Assume to the contrary that the right child $\beta = [i_\beta..j_\beta]$ of $\alpha_\ell$ is not $\lambda = [i_\lambda..j_\lambda]$.
  By definition of  $\alpha_\ell$, $\beta$ and $\lambda$,
  we have that $\beta$ must be an ancestor of $\lambda$ and
  $i^\prime \leq i < i_\beta < i_\lambda$ since $\lambda$ is a right node.
  Also, it must be that $j \leq j^\prime < j_\beta$ since otherwise,
  $w[i_\beta..j_\beta]$ would have period $p^\prime < |[i_\beta..j_\beta]|$ due to run $r$, contradicting that it is a Lyndon word.
  However, by the definition of $\ell$, this implies that  $w[i_\lambda..i_\beta] \lexorderl{\ell} w[i_\beta..i_\beta]$,
  still contradicting that $w[i_\beta..j_\beta]$ is a Lyndon word.

  Therefore, if $\exrun([i..j])$ exists, 
  we can find it by checking the two nodes that are the right children
  of $\alpha_0$ and $\alpha_1$ in constant time.
\end{proof}

\section{Conclusion}
\label{sec:conclusions}
We show a remarkably simple proof to the 15 year-old runs conjecture,
by discovering a beautiful connection between the \Lroot{}s of runs and
the longest Lyndon word starting at each position of the string.
We also show a bound of $\sigma(n) < 3n$ for the maximum sum of exponents 
of runs in a string of length $n$, improving on the previous best bound of $4.1n$~\cite{crochemore12},
as well as improved analyses on related problems.
We also proposed a simple linear-time algorithm for computing all the runs in a string.
Furthermore, realizing that the longest Lyndon word starting at each position of the string corresponds
to a right node in the Lyndon tree, we showed a simple optimal solution to the 2-Period Query problem.

The characterizations of runs in terms of Lyndon words as shown in this paper
significantly improves our understanding of how runs can occur in strings.
A remaining question is the exact value of $\lim_{n\rightarrow\infty} \rho(n)/n$,
which is known to exist but is never reached~\cite{giraud09:_asymp}.

\section*{Acknowledgments}
HB,SI,MT were supported by JSPS KAKENHI Grant Numbers
25280086, 26280003, 25240003.
The authors thank (in alphabetical order)
Maxime Crochemore,
Antoine Deza,
Frantisek Franek,
Gregory Kucherov,
Simon Puglisi, and
Ayumi Shinohara
for helpful comments and discussions.

\bibliographystyle{siam}

\end{document}